\documentclass[letterpaper, 10 pt, conference]{ieeeconf}  % Comment this line out
\usepackage{amsmath,amssymb,amsfonts}
\usepackage{graphicx}
\usepackage{textcomp}
\usepackage{mathtools}

\usepackage{array}
\usepackage[latin1]{inputenc}
\usepackage[T1]{fontenc}

\usepackage{siunitx}
\usepackage{subeqnarray}
\usepackage[latin1]{inputenc}
\usepackage[T1]{fontenc}
\usepackage{amsmath,amsfonts,amssymb}
\usepackage{indentfirst}
\usepackage{graphicx}
\usepackage{remark}
\usepackage{subeqnarray}
\usepackage{pifont}
\usepackage{ass}
\usepackage{xcolor}
\usepackage{fancybox}
\usepackage{psfrag}
\usepackage{float}
\usepackage[thinlines]{easymat} 
\usepackage[nocross]{pmat}
\usepackage{mdwlist} 
\usepackage[percent]{overpic}
\usepackage{epstopdf}
\usepackage{smallmatrices} 
\usepackage{rotating}
\usepackage{cases}
\definecolor{green2}{RGB}{0,128,0}
 \usepackage{latexsym}
\usepackage[normalsize,it,hang]{subfigure}
\usepackage{setspace}
\usepackage{subfigure}
\usepackage{balance}

\usepackage{tikz}
\newcommand*\circled[1]{\tikz[baseline=(char.base)]{
            \node[shape=circle,draw,inner sep=2pt] (char) {#1};}}
            
\usepackage[symbol]{footmisc}

\usepackage{cite}

\renewcommand{\dblfloatsep}{0pt} % espace entre flottants et d�but/fin
\renewcommand{\dbltextfloatsep}{0pt} % espace entre flottants du d�but/fin
\renewcommand{\intextsep}{0pt} % espace vertical entre flottants et
\allowdisplaybreaks % Autorise la coupure de fin de page

\def\aujour{\number\day \space \ifcase\month\or
janvier\or f�vrier\or mars\or avril\or mai\or
juin\or juillet\or ao�t\or septembre\or octobre\or
novembre\or d�cembre\fi \space \number\year}

\def\cH{{\cal H}}

\def\cL{{\cal L}}

\newtheorem{remark}{Remark}

\newtheorem{defi}{Definition}
\newtheorem{prob}{Problem}
\newtheorem{prop}{Proposition}
\def\C{{\setbox0=\hbox{$\displaystyle{\rm C}$}
        \hbox{\hbox to0pt{\kern 0.4\wd0\vrule height 0.95\ht0\hss}\box0}}}
\def\Q{{\setbox0=\hbox{$\displaystyle{\rm Q}$}%
    \hbox{\raise 0.2\ht0\hbox to0pt{\kern 0.4\wd0\vrule height
    0.85\ht0\hss}\box0}}} % ensemble des Rationnels --> � revoir
\def\R{\mathop{\rm I\mkern -3.5mu R}} % ensemble des R�els --> OK
\def\cH2{{\cal H}_2} % H2 calligraphi� (norme)

\def\cL2{\mathop{\mathcal L}_{2}} % espace carr�-int�grable
\def\cRH2{\mathop{\cal R \cal H}_2} % RH2 calligraphi� (norme)
\def\cRL2{\mathop{\cal R \cal L}_{2}} % espace carr�-int�grable

 \def\var{\mbox{\bf Var}}

\DeclareMathOperator*{\der}{d}

\newcommand{\norm}[1]{\left\|{#1}\right\|}
\newcommand{\abs}[1]{\left|{#1}\right|}

\makeatletter
\DeclareRobustCommand\sfrac[1]{\@ifnextchar/{\@sfrac{#1}}
                                            {\@sfrac{#1}/}}
\def\@sfrac#1/#2{\leavevmode\kern.1em\raise.5ex
         \hbox{$\m@th\fontsize\sf@size\z@
                           \selectfont#1$}\kern-.1em
         /\kern-.15em\lower.25ex
          \hbox{$\m@th\fontsize\sf@size\z@
                            \selectfont#2$}}
\makeatother
\IEEEoverridecommandlockouts                              % This command is only
\title{\LARGE \bf
Tracking Control of Optical Beam Transceivers using Mean Field Models}
\author{Ibrahima N'Doye$^1$ and Taous-Meriem~Laleg-Kirati$^{1,2}$
\thanks{This work has been supported by the King Abdullah University of Science and Technology (KAUST), Base~Research~Fund~under Grant (BAS/1/1627-01-01).}
\thanks{$^{1}$Computer, Electrical, and Mathematical Sciences and Engineering Division, King Abdullah University~of~Science~and~Technology (KAUST),~Saudi~Arabia. 
        {\tt ibrahima.ndoye@kaust.edu.sa}, {\tt taousmeriem.laleg@kaust.edu.sa}.}%
\thanks{$^{2}$National Institute for Research in Digital Science and Technology, Paris-Saclay, France.}
}

\begin{document}

\maketitle
\thispagestyle{empty}
\pagestyle{empty}

%%%%%%%%%%%%%%%%%%%%%%%%%%%%%%%%%%%%%%%%%%%%%%%%%%%%%%%%%%%%%%%%%%%%%%%%%%%%%%%%
\begin{abstract}
This paper proposes mean field models to maintain an accurate line-of-sight and tracking between transceivers mounted in mobile unmanned aerial vehicles (UAVs) platforms in the presence of underlying mechanical vibration effects. We describe a two-way optical link beam tracking control that coordinates mobile UAVs in a network architecture to provide reliable network structure, distributed connectivity, and communicability, enhancing terrestrial public safety communication systems. We derive the optical transceiver trajectory tracking problem in which each agent dynamic and cost function is coupled with other optical beam transceiver agent states via a mean field term. We propose two optimal mean field beam tracking control frameworks through decentralized and centralized strategies in which the optical transceivers compete to reach a Nash equilibrium and cooperate to attain a social optimum, respectively. The solutions of these strategies are derived from forward-backward ordinary differential equations and rely on the linearity Hamilton-Jacobi-Bellman Fokker-Planck (HJB-FP) equations and stochastic maximum principle. Moreover, we numerically compute the solution pair of the resulting joint equations using Newton and fixed point iteration methods to verify the existence and uniqueness of the equilibrium and social optimum.
\end{abstract}

%%%%%%%%%%%%%%%%%%%%%%%%%%%%%%%%%%%%%%%%%%%%%%%%%%%%%%%%%%%%%%%%%%%%%%%%%%%%%%%%
\section{Introduction}
Designing decision rules for distributed autonomous systems to ensure an objective in dynamic games has gained significant attention in various aspects. Distributed autonomous (DA) systems have opportunities for several applications, including patrolling in communication-restricted environments, traffic modeling, rescue and search, scientific data collection, cargo delivery, object localization, and surveillance of disaster areas \cite{GKBFMMG:18,SPCK:18,MaS:15}.  However, such DA systems require cooperation and coordination capabilities at multiple levels to avoid collisions \cite{GKBFMMG:18,SPCK:18}. In contrast to the standard centralized control, the distributed control design of DA systems should coordinate globally based on limited local information~\cite{MaS:15}. For instance, the consensus and synchronization of unmanned aerial vehicles (UAVs) based on optical wireless channel communication that benefits from mobility and optical data rates are important in public safety communications.% 

Optical wireless communication (OWC) technology is essential in several applications in which a congested spectrum reduces radio-frequency (RF) communication technology and fails to provide available data  \cite{MaR:10,EMH:11,BBDHY:12,GPR:12}.
As the demand for capacity has grown over the past years, it has become a promising communication technology due to its flexibility in carrying high-speed and cost-effective networking solutions \cite{GPR:12,ZNBAAL:20,NZAL:21}. In addition, it provides high data rates, low cost and power consumption, and low latency~\cite{HaR:08,HTOWF:11,FLMS:09}. Although significant efforts are undertaken to design reliable OWC for mobile networking sensing applications, the misalignment between OWC mobile platforms remains a considerable challenge due to the underlying mechanical vibration and accidental shocks. A misalignment can result in optical channel disconnection, leading to connectivity loss. Subsequently, the needed alignment orientation angles are not directly accessible from measurements and must be estimated. The authors in \cite{NZAL:21,NZAZRL:22} have recently developed a novel nonlinear estimator for a light-emitting-diode-based optical communication model based on switched gains to overcome the estimation problem. However, the distributed information suffers from time delay in a large-scale system, impacting the optical communication networks. Hence, local decisions based on partial information must be considered to coordinate globally in a dynamically evolving setting.

 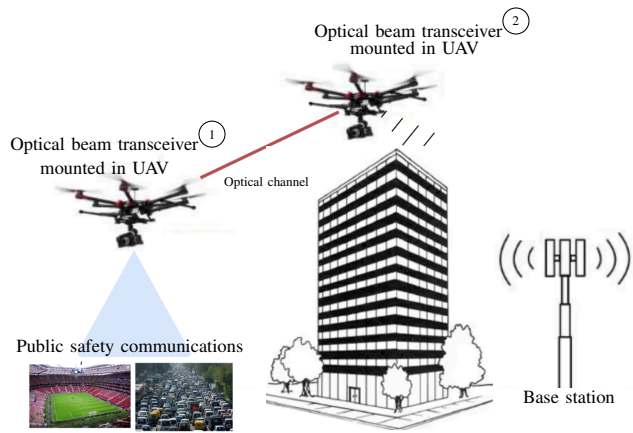
\begin{figure}[!t]
 \vskip 0.3in  
\centering
      \begin{overpic}[scale=0.2]{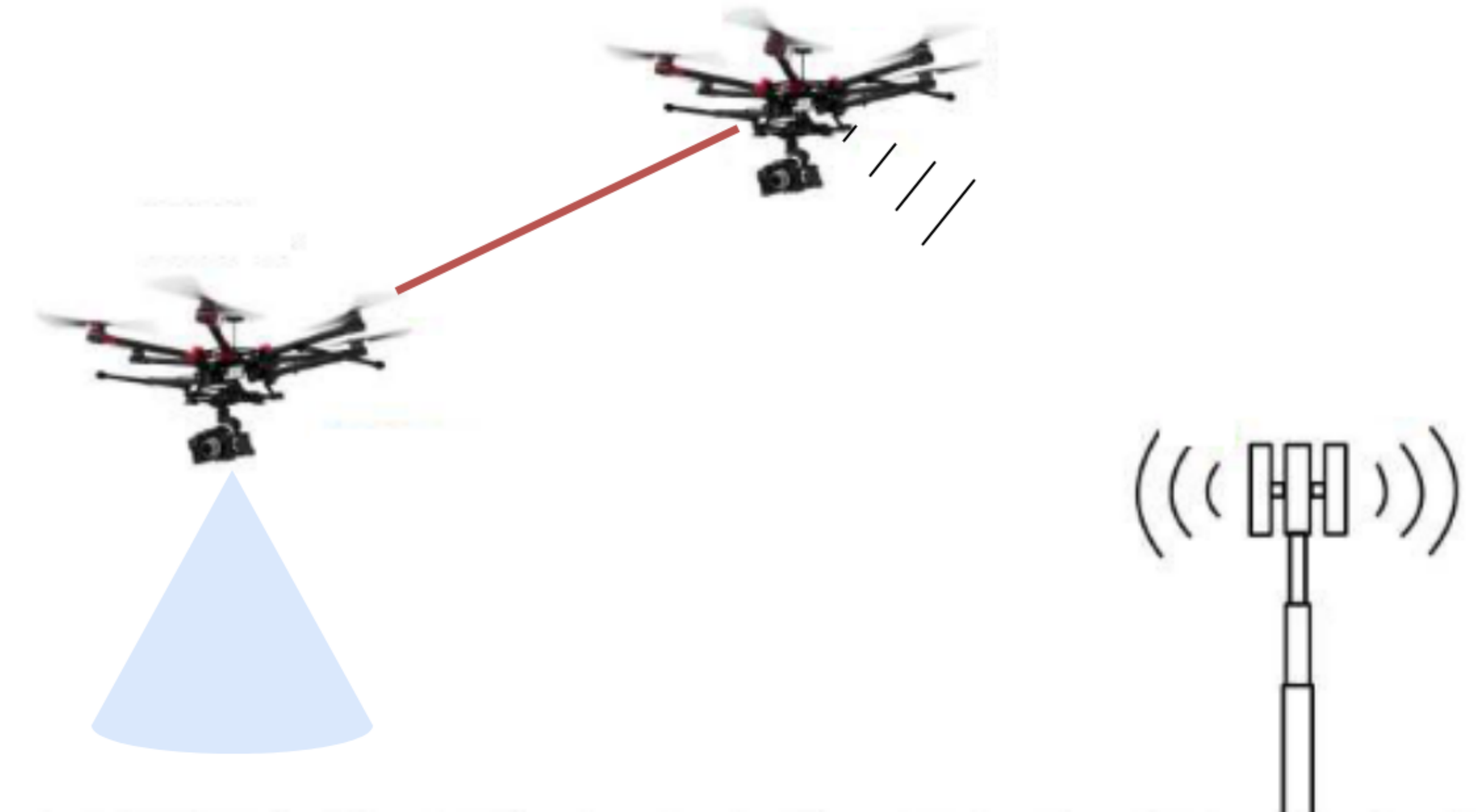}
       \put(81,-2.5){\scriptsize Base station}
               \put(46,59){\scriptsize Optical beam transceiver$^{\tiny\circled{2}}$}
                 \put(52,56){\scriptsize mounted in UAV}
         %\put(52,56){\scriptsize UAV$^{\tiny\circled{2}}$}
              \put(-5,40){\scriptsize Optical beam transceiver$^{\tiny\circled{1}}$}
               \put(0,36){\scriptsize mounted in UAV}
        % \put(12,38){\scriptsize UAV$^{\tiny\circled{1}}$}
           \put(-4,6){\scriptsize Public safety communications}
       \put(38,-7){\begin{overpic}[scale=0.18]{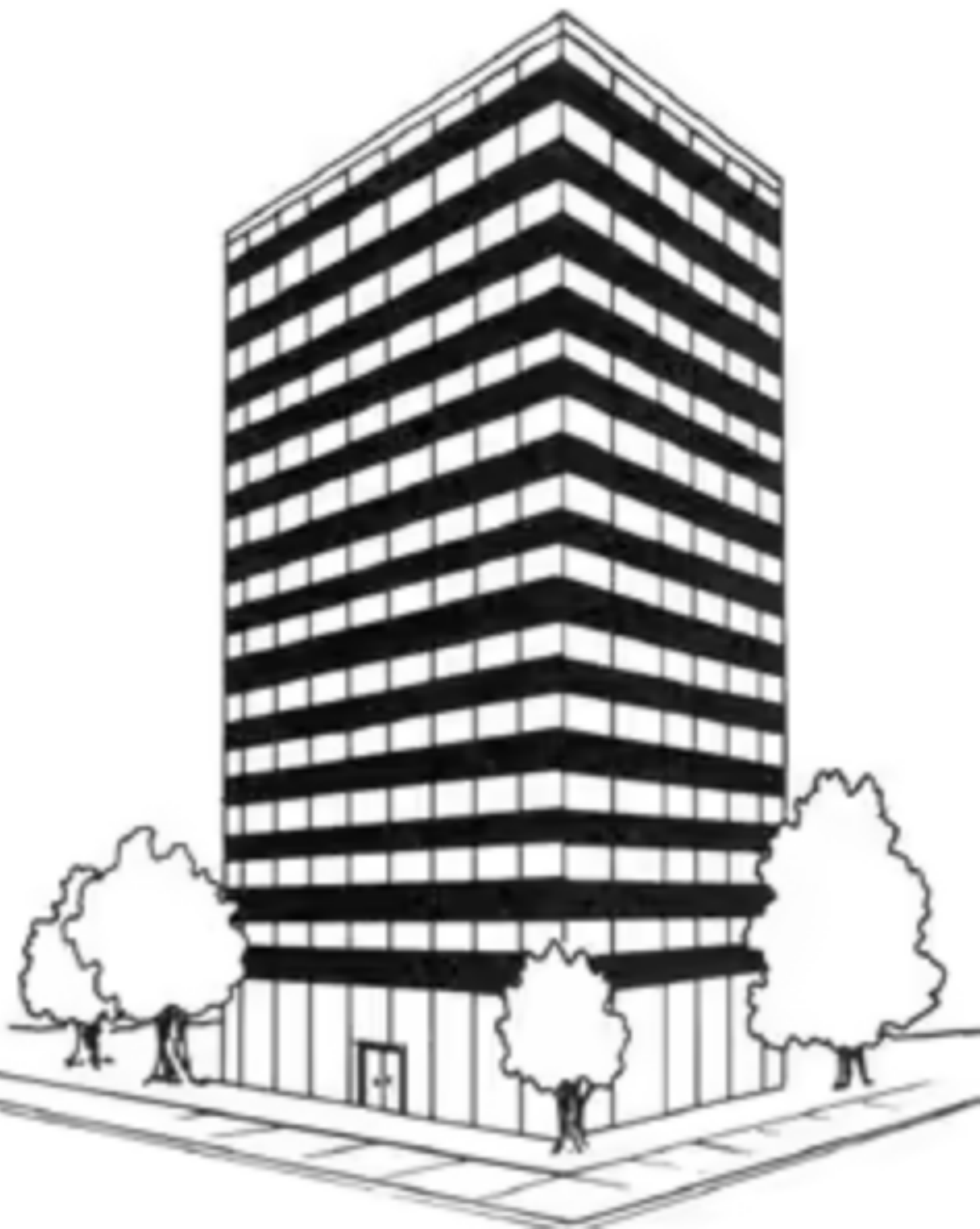} %height=2.4cm,width=6.4cm
          \put(-15,86){\tiny Optical channel}
  \end{overpic}}
         \put(-2,-7.5){\begin{overpic}[scale=0.15]{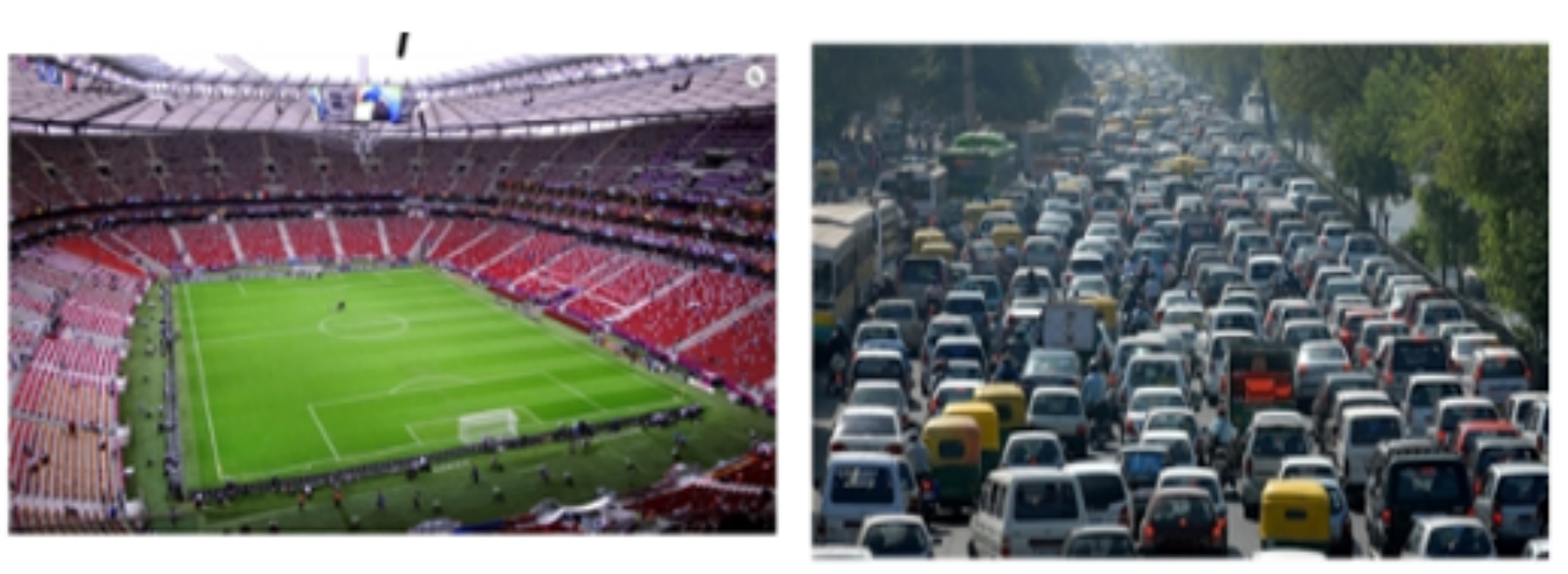} %height=2.4cm,width=6.4cm
  \end{overpic}}
            \end{overpic}  \vspace{10pt}
              \caption{Hybrid radio frequency/optical beam transceivers mounted in UAVs to assist public safety communications.}\label{fig2a}
              \vskip -0.2in    
 \end{figure}

In the context of optical beam consensus for mobile stations in public safety communications (see Fig. \ref{fig2a}, for example), our goal is to cooperate, perform missions, and share information through the optical communication link or hybrid optical/radio-frequency communication for long distances among the unmanned aerial vehicles (UAVs) in harsh environments. Note that the pointing and spatial acquisition achieve the coarse alignment operations before the fine alignment control. Hence, this fine alignment operation compensates precisely for the continuous relative motion of the UAVs and performs optical beam tracking or cooperative beam tracking and data transmission \cite{LoY:87}.

In doing so, the primary objective is to continuously track the arrival direction of the incident beam and transmit the beam back in that direction. At each transceiver mounted in a UAV, a position-sensitive photodetector measures the azimuth and elevation angles of the beam tracking error, which is the displacement of the beam's arrival direction to the direction normal to the receiving aperture. Hence, the heading of the optical transmitting device (i.e., azimuth and elevation directions) is controlled according to the measured beam tracking error. Indeed, the agents seek to align their azimuth and elevation angular orientations via a servo-driven pointing assembly. To promote the interactions of distributed decision rules between transceivers within the optical beam, we study the fine alignment control problem from a mean field game theory.

Mean field games and control have emerged as potential frameworks for controlling behavior in many fields, including distributed system control, social science, computer science, engineering, and economics \cite{BFP:13,GoS:14,BDT:20}. Mean field games (MFGs) rely on characterizing Nash equilibria under some rationalities conditions \cite{MoB:14}. Since the state-space dimension and the complexity can increase with the number of players in a differential game, this characterization may be impractical \cite{HCM:07,MoB:14,AJWG:08}. Besides, the fact that each agent may only access local information in a decentralized framework adds another challenge. The authors in \cite{HMC:06,LaL:07,HCM:07} proposed a mean field game framework to tackle the scalability problem in multi-agent systems. The typical feature of the MFG is that the game problem can be modeled as a stochastic optimal control problem. The most common forms of MFG that arise in several applications are linear quadratic MFGs (LQ-MFGs) (see, \cite{MoB:14,HCM:12,HCM:07,SMN:20} and references therein). 

A fundamental difficulty in the MFG problem is finding the mean field equilibrium, which constitutes the bottleneck of numerical methods for MFG and mean field control problems. In this paper, we formulate two linear quadratic mean field  optical beam tracking models adopting the HJB-FP approach to numerically compute the existence and uniqueness of equilibrium using Newton iteration and fixed point iteration methods. The HJB-FP approach does not need to use the Riccati equation to solve the optimal trajectory and can benefit from setting MGFs problems (see, for instance, ~\cite{LaW:11}).

The paper is organized as follows. Section \ref{sec-model} describes the stochastic mean field optical beam tracking model, including the optical communication link model and the line-of-sight (LoS) model. In Section \ref{sec-pb}, we formulate the mean field stochastic optical beam tracking problem. In Section \ref{sec-solution}, we provide the mean field game and mean field control solutions based on the HJB-FP equation method. Simulation results are provided in Section \ref{num-simu} to show the mean field equilibrium and social optimum solutions. Section \ref{concl} concludes the paper. 

{\bf Notation:}  Subscript $t$ denotes the time dependence. Matrix $A^{\top}$ is the transposed matrix of $A$. The identity matrix of dimension $r$ is denoted $I_{r}$.  $\var$ and $ \mathbb{E}$ are the variance and expectation, respectively. $\cL2$ is defined as $\norm{\xi(t)}\!\!=\!\!\left(\int_0^t \abs{\xi(\tau)}^2\der\!\tau\right)^\frac{1}{2}\!\!$ and $\abs{\xi(\tau)}$ is the vector norm of $\xi$. %
%%%%%%%%%%%%%%%%%%%%%%%%%%%%%%%%%%%%%%%%%%%%%%%%%%%%%%%%%%%%%%%%%%%%%%%%%%%%%%%%%
%
\section{Stochastic optical beam tracking model}\label{sec-model}
This section describes the beam tracking and line-of-sight optical beam models that comprise the optical beam transceivers mounted on indistinguishable unmanned aerial vehicles (UAVs).

\subsection{Optical link model}\label{sec-opt-link}
We consider a two-way optical source and destination channel communication mounted on UAVs with identical transceivers. Specifically, the architecture can adapt complete hybrid optical/radio-frequency (RF) communication, where each transceiver drives its optical beam to the opposite and receives the optical beam of the opposing transceiver, as illustrated in Fig.~\ref{fig1a}. Following this architecture, the decisions are coupled and distributed among the optical transceivers~\cite{GMRZ:20}. %

 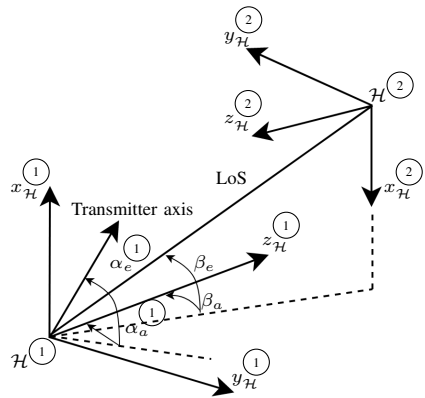
\begin{figure}[!h]
      \vskip 0.2in  
\centering
      \begin{overpic}[scale=0.18]{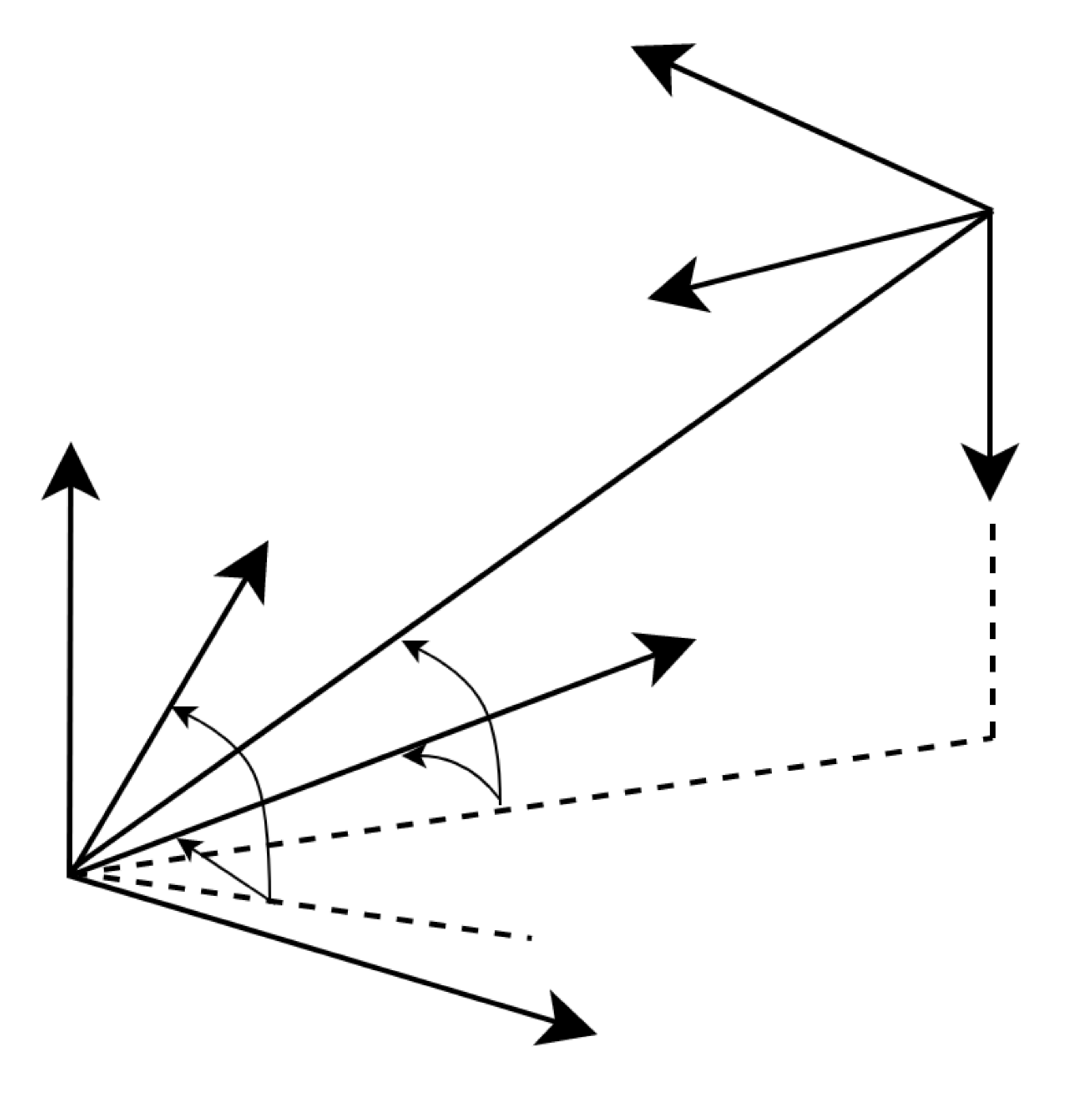}
        \put(-4,12){\scriptsize $\mathcal{H}^{\tiny\circled{1}}$}
        \put(-4,58){\scriptsize $x_\mathcal{H}^{\tiny\circled{1}}$}
         \put(54,8){\scriptsize $y_\mathcal{H}^{\tiny\circled{1}}$}
         \put(62,44){\scriptsize $z_\mathcal{H}^{\tiny\circled{1}}$}
           \put(26,21){\scriptsize $\alpha_a^{\tiny\circled{1}}$}
              \put(22,38){\scriptsize $\alpha_e^{\tiny\circled{1}}$}
        \put(90,82){\scriptsize $\mathcal{H}^{\tiny\circled{2}}$}
        \put(94,58){\scriptsize $x_\mathcal{H}^{\tiny\circled{2}}$}
          \put(52,98){\scriptsize $y_\mathcal{H}^{\tiny\circled{2}}$}
           \put(52,76){\scriptsize $z_\mathcal{H}^{\tiny\circled{2}}$}
                      \put(46,28){\scriptsize $\beta_a$}
              \put(44,38){\scriptsize $\beta_e$}
           \put(12,52){\scriptsize Transmitter axis}
           \put(50,60){\scriptsize LoS}
            \end{overpic}  \vspace{-10pt}
              \caption{Coordinate systems, line-of-sight (LoS), and angles of the transceiver axis. The subscripts $e$ and $a$ stand for elevation and azimuth, respectively. }\label{fig1a}
              \vskip 0.1in    
 \end{figure}
 
We assume that the optical transceivers are subject to relative motion. To refer to the transceivers, we use the superscripts $\tiny\circled{1}$ and $\tiny\circled{2}$ or $k=\tiny\circled{1},\tiny\circled{2}$.  The transceivers are mounted on the UAV platform, each composed of a laser source, a position-sensitive photodetector (PSP), and a lens. The PSP center is at the lens's focus, and its surface remains perpendicular to the lens axis. Hence, enabling the laser source and lens axes to be parallel to the transceiver axis. Consequently, we can maintain the elevation and azimuth of the transceiver axis using an electromechanical pointing device. 

We consider the centered and non-rotating right and left-handed coordinate systems $\mathcal{H}^{\tiny\circled{1}}$ and $\mathcal{H}^{\tiny\circled{2}}$, respectively.  As shown in Fig.~\ref{fig1a}, the fixed coordinate axes are parallel and in opposite directions. Now, consider the two-dimensional vector  
$\alpha^{k}_t$ be the azimuth and elevation angles of the transceiver axis $k$ to the coordinate system $\mathcal{H}^k$, and $\beta_t$ are azimuth and elevation angles of the LoS to the coordinate system $\mathcal{H}^k$. Then, the beam tracking error at the transceiver $k$ mounted in the UAV is defined as follows~\cite{KKN:07}
\begin{equation}\label{eq1a}
\theta^{k}_t=\alpha^{k}_t-\beta_t, 
\end{equation}
which is equivalent to the alignment error of the opposite transceiver platform. As $\beta$ is equivalent for the transceiver platforms due to the definitions of $\mathcal{H}^k$, we drop the superscript $k$.

Assume that the transmitted optical fields drive into the receivers along the line-of-sight (LoS) is independent of the pointing error of the transmitters. As a result, the received optical field image describes a spot of light over the surface of the receiver aperture with a bell-shaped intensity profile $\rho_t(\ell)$ centered at $f_r\theta^{k}_t$ where $\ell$ is the position vector of a point on the surface of the photodetector and $f_r$ denotes the lens's focal length~\cite{KKN:07}. In the absence of pointing error induced by transceiver$^{\tiny\circled{2}}$, we denote by $\mathcal{P}_t^1\!>\!0$ the total optical power received by transceiver$^{\tiny\circled{1}}$ at time $t$. Therefore, the optical intensity over the surface of the receiver aperture ${\tiny\circled{1}}$ is given as follows
\begin{equation}\label{eq2a}
\mathcal{I}_t^{\tiny\circled{1}}(\ell)=\mathcal{P}_t^{\tiny\circled{1}}\underbrace{\exp\left(-\frac{2\norm{\theta_{t}^{\tiny\circled{2}}}^2}{\tilde{\varphi}^2}\right)}_{\mbox{\scriptsize attenuation factor}}\rho_t\left(\ell-f_r\theta^{{\tiny\circled{1}}}_t\right),
\end{equation}
where $\tilde{\varphi}$ is the divergence angle induced by transceiver$^{\tiny\circled{2}}$ through a circularly symmetric Gaussian beam. The platforms mounted on both transceiver$^{\tiny\circled{1}}$ and transceiver$^{\tiny\circled{2}}$ work as transceivers; we derive similar expressions of $\mathcal{I}_t^b(\ell)$ by flipping ${\tiny\circled{1}}$ and ${\tiny\circled{2}}$ in \eqref{eq2a}. As $\theta^{k}_t$ is linearly dependent on the displacement of the spot of light $\ell$, we can then estimate $\theta^{k}_t$ from the output of the PSP for a given small pointing error of transceiver$^{\tiny\circled{2}}$\!. Consequently, a control signal can drive the beam tracking error to zero through the {\it pointing mechanism\footnote{\scriptsize{Pointing mechanism is an electromechanical system that needs to be adjusted accurately to correct the pointing error of the optical transceiver signal}}} from this estimate leading to maximum optical power.

\begin{table*}[!t]
\begin{center}
\line(1,0){480}
\end{center}
\begin{align}
\mathcal{J}^k(u)=\mathcal{J}^k(u^{k},u^{-k})
 \overset{\Delta}{=}&\frac{1}{2}\mathbb{E}\displaystyle\left[\int_0^T\!\!\!\Big((x_t^{k})^\top Q_tx_t^{k}+(u_t^{k})^\top R_tu_t^{k}\Big)\der\!t +(x_T^{k})^\top Q_Tx_T^{k}\right]
\notag \\&+
\frac{1}{2}\mathbb{E}\displaystyle\left[\!\int_0^T\!\!\!\Big(x_t^{k}\!-\!S_t\frac{1}{N\!-\!1}\!\!\sum_{i=1,i\neq k}^N\!\!\!x_t^{i}\Big)^\top \!\bar{Q}_t\Big(x_t^{k}-S_t\frac{1}{N\!-\!1}\!\sum_{i=1,i\neq k}^N\!\!\!x_t^{i}\Big)\!\der\!t \!\right]
\notag \\&+
\frac{1}{2}\mathbb{E}\displaystyle\left[\!\Big(x_T^{k}\!-\!S_T\frac{1}{N\!-\!1}\!\!\sum_{i=1,i\neq k}^N\!\!\!x_T^{i}\Big)^\top\!\bar{Q}_T\Big(x_T^{k}-S_T\frac{1}{N\!-\!1}\!\sum_{i=1,i\neq k}^N\!\!\!x_T^{i}\Big)\!\right], \qquad 1\leqslant k\leqslant N \tag{$\star$}
\end{align}
\begin{center}
\line(1,0){480}
\end{center}
\vspace{0.25cm}
\end{table*}

\subsection{Mean field model for optical beam tracking}\label{sec-syst}
We propose a flying ad-hoc network architecture to coordinate optical wireless communication between all transceivers mounted in the UAVs in a network. This infrastructure features complete optical beam transceiver communication for long distances, where each transceiver agent $k$ is influenced only by its neighboring transceiver agents, as illustrated in Fig.~\ref{fig3a}.

 \begin{figure}[!t]
\centering
      \begin{overpic}[scale=0.15]{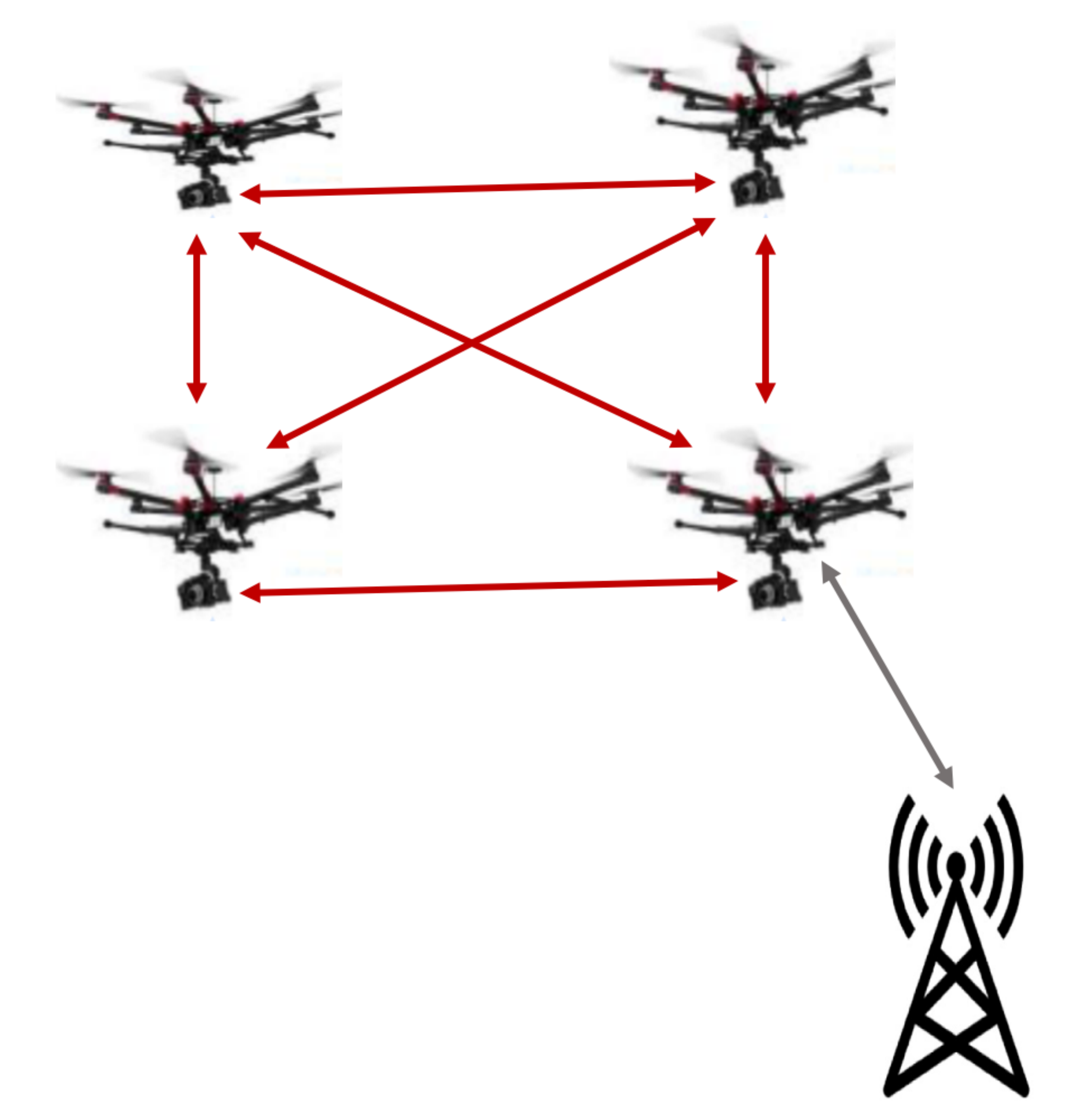}
                     \put(80,59){\tiny Optical beam transceiver}
                 \put(82,54){\tiny mounted in UAV}
                     \put(80,89){\tiny Optical beam transceiver}
                 \put(82,84){\tiny mounted in UAV}
                     \put(-30,59){\tiny Optical beam transceiver}
                 \put(-28,54){\tiny mounted in UAV}
                     \put(-30,89){\tiny Optical beam transceiver}
                 \put(-28,84){\tiny mounted in UAV}
                          \put(34,77){\tiny \textcolor{red}{Optical link}}
       \put(74,-2.5){\scriptsize Base station}
         \put(-12,-6){\begin{overpic}[scale=0.24]{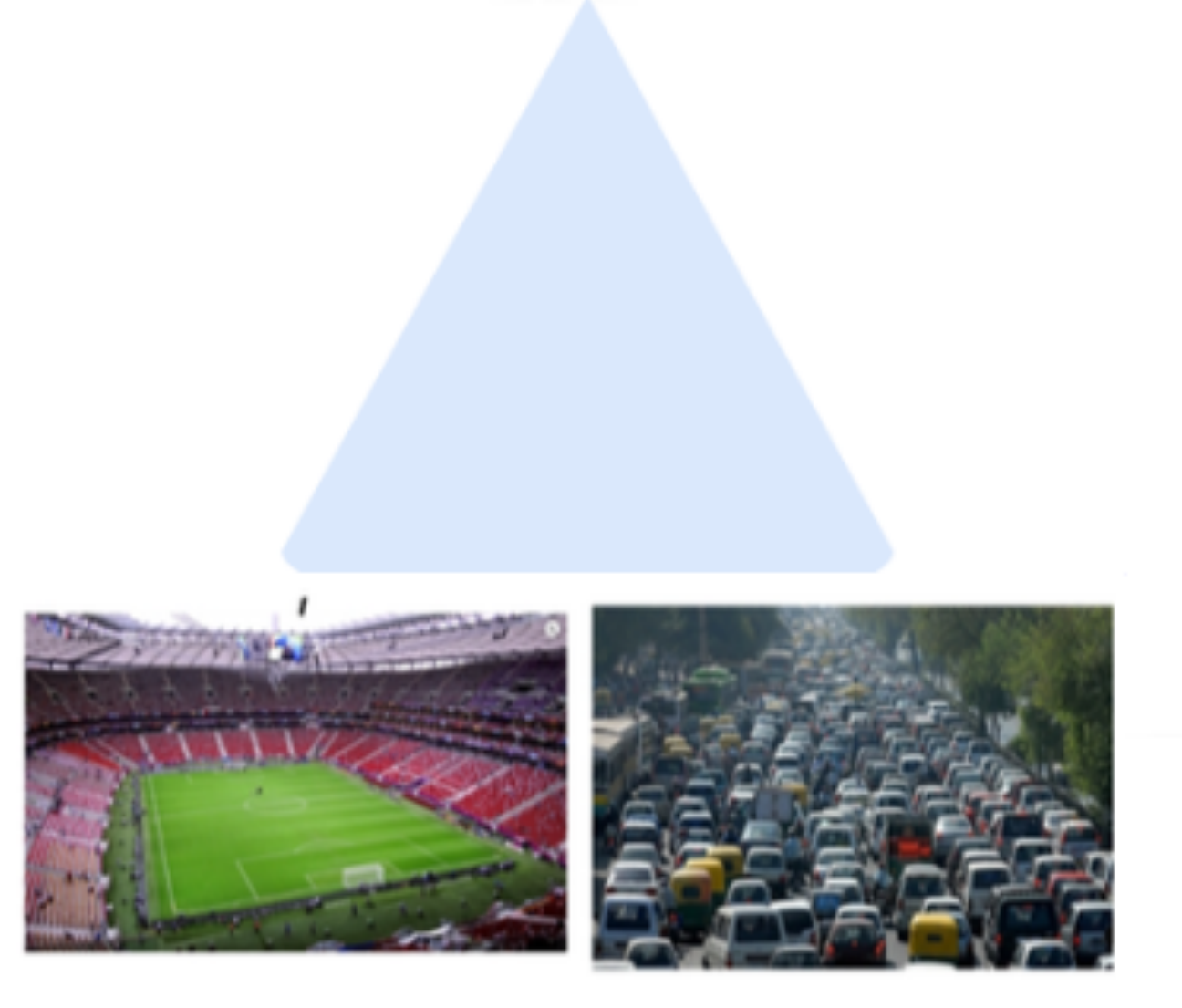} %height=2.4cm,width=6.4cm
                 \put(-4,46){\scriptsize Public safety communications}
                \put(64,78){\scriptsize Network architecture}
  \end{overpic}}
            \end{overpic}  \vspace{-2pt}
              \caption{$N$ mobile optical beam transceivers (four in this example) mounted in UAVs network based on flying ad-hoc network architecture in public safety communications operations.}\label{fig3a}
              \vskip -0.2in    
 \end{figure}
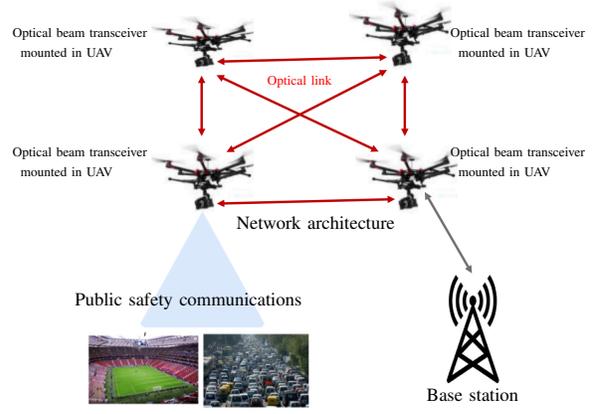 
 The pointing mechanism is an electromechanical system that involves electromechanical motion devices between transceivers within a network architecture. It does not include the UAV dynamics. Let the input torque vector $u_t^k\in \R^2$ and the output vector $\theta_t^k\in \R^2$ be the two-dimensional azimuth and elevation angles.

We describe the pointing assembly for the fine control regime by the following linear stochastic differential equation~\cite{KKN:07}
\begin{equation}\label{eq3a}
\left\{\begin{array}{llllllll}
\!\!\!\der\!\bar{x}_t^{k}\!=\!\displaystyle \Big[\bar{A}_t \bar{x}_t^{k}+\bar{\Gamma}_t \bar{x}_t^{N}+\bar{B}_t \bar{u}_t^{k}\Big]\!\!\der\!t+\bar{D}_t \der\!\mathcal{\bar{B}}_t^{k}\\
\alpha^{k}_t=\bar{C}_t \bar{x}_t^{k}
\end{array}\right.
\end{equation}
where $\bar{x}_t^{k}\!\in\!\R^{\bar{n}}$ is the state vector, $\mathcal{\bar{B}}_t^{k}\!\in\!\R^{\bar{q}}$ is a $\bar{q}$ dimensional independent Wiener processes induced by the movement of the underlying transceiver optical platform, and $\bar{A}_t$, $\bar{B}_t$, $\bar{C}_t$, $\bar{D}_t$ and $\bar{\Gamma}_t$ are uniformly bounded matrices with appropriate dimensions. $N$ is the number of transceivers, $\bar{x}_t^{N}=\displaystyle\frac{1}{N\!-\!1}\!\!\sum_{i=1,i\neq k}^N\!\!\!\bar{x}_t^{i}$ is the coupled terms across the transceivers that are induced from the distributed nature of the optical wireless communication architecture. This term introduces the mean field interactions into the optical beam model. Transceiver $k$ mounted in a UAV controls their positions by moving their driving torques.

We model the LoS $\beta_t$ by the following linear stochastic differential equation
\begin{equation}\label{eq4a}
\left\{\begin{array}{llllllll}
\!\!\der\!\tilde{x}_t\!=\!\displaystyle \tilde{A}_t \tilde{x}_t\!\der\!t\!+\!\tilde{D}_t \der\!\mathcal{\tilde{B}}_t\\
\beta_t=\tilde{C}_t \tilde{x}_t
\end{array}\right.
\end{equation}
where $\tilde{x}_t\!\in\!\R^{\tilde{n}}$ is the state vector, $\mathcal{\tilde{B}}_t\!\in\!\R^{\tilde{q}}$ is a $\tilde{q}$ dimensional standard Brownian motions, and $\tilde{A}_t$, $\tilde{C}_t$, and $\tilde{D}_t$ are uniformly bounded matrices with appropriate dimensions. 

From the linear combination of the position vector $f_r\theta^{k}_t$ resulting from the sum of $\bar{x}_t^{k}$ and $\tilde{x}_t$, we combine \eqref{eq3a} and \eqref{eq4a} as an augmented state $x_t^{k}=\begin{bmatrix}  \bar{x}_t^{k}\\\tilde{x}_t\end{bmatrix}$ where $x_t^{k}\!\in\!\R^{n}$ is the state vector with $n=\bar{n}+\tilde{n}$. Then, the state process $x_t^{k}$ of the $k$th transceiver agent resulting of the tracking error is given in a compact form as follows
\begin{equation}\label{eq5a}
\left\{\begin{array}{llllllll}
\der\!x_t^{k}=\displaystyle \Big[\mathbb{A}_t x_t^{k}+\mathbb{B}_t u_t^{k}+\Gamma_t x_t^{N,k}\Big]\der\!t+\mathbb{D}_t \der\!\mathcal{B}_t^{k}\\
\theta^{k}_t=\mathbb{C}_t x_t^{k}
\end{array}\right.
\end{equation}
We assume that the initial state $x_0^{k}$ is a Gaussian vector with mean $\hat{x}_0$ and covariance matrix $\sigma_0^i$ and independent
of $\big\{\mathcal{B}_t^{\tiny\circled{1}}, t\geqslant0\big\}$ and $\big\{\mathcal{B}_t^{\tiny\circled{2}}, t\geqslant0\big\}$.

The resulting ad-hoc network-based architecture of the optical beam model relies on the fact that a single transceiver mounted in UAVs and the base station only have access to the statistical distributions of the azimuth and elevation angles of a single optical transceiver mounted in UAVs. In line with this, we propose solving a mean field game problem through an optimal decentralized strategy and a mean field control problem via an optimal centralized.

Throughout the paper, we will consider the dynamics of the $N-$transceiver players, $1\leqslant k\leqslant N$ and set $\mathbb{D}_t$ a constant parameter with $\rho=\frac{1}{2}\mathbb{D}_t^\top\mathbb{D}_t$.

\section{Problem statement}\label{sec-pb}
This section considers two feedback states information: Each optical transceiver player observes $x_t^{k}$ and seeks Nash strategies individually in the decentralized state feedback setting, while in the centralized framework, the optical transceivers cooperate to reach a social optimum. 

\subsection{Mean field game problem}\label{MFG-sec-prob}
If we recall the mean field game formulation, the primary objective of each transceiver player is to minimize its own cost functional by properly controlling its dynamics.
Let $u_t = (u_t^1,\cdots,u_t^N)$ denote the set of the optical transceiver agents' control strategies and $u_t^{-k}= (u_t^1,\cdots,u_t^{k-1}, u_t^{k+1},\cdots,u_t^N)$ denote the control strategies set except $\mathcal{A}_k$ $(\mathcal{A}_k, 1\leqslant k\leqslant N)$. Then, $\mathcal{A}_k$'s expected cost function under a given joint control $u_t=(u_t^{k},u_t^{-k})$ is defined in ($\star$) where $T>0$ is the time horizon, $Q_t$, $R_t$,  $\bar{Q}_t$ are deterministic positive definite bounded matrix-valued functions in time with appropriate dimensions and $S_t$ is bounded and deterministic matrix-valued function in time with suitable dimension. \textcolor{black}{The first expectation describes each platform module's ({\it i.e.,} each player itself) running expenses and terminal costs. The other two expectations penalize the deviation error performance from the average behavior.}

In this setting, we approximate the state-average term in ($\star$) by a mean field trajectory, also known as deterministic sequence according to the Nash certainty equivalence principle \cite{HMC:06}. Hence, this problem reduces to finding an equilibrium for the optical beam fine-tracking in which we will design a state feedback-based strategy $u_t^{k}$ together with $u_t^{-k}$ to ensure a unique solution of $x_t^{k}$ on $[0,\, T]$.
\begin{defi}
A set of strategies $(\tilde{u_t}^1,\cdots,\tilde{u_t}^N)$ is a Nash equilibrium if for all $1\leqslant k\leqslant N$, the following comparison inequality is satisfied
\begin{equation}\label{eq6a}
\mathcal{J}^k(\tilde{u_t}^{k},\tilde{u_t}^{-k})\leqslant \mathcal{J}^k(u_t^{k},\tilde{u_t}^{-k}),
\end{equation}
for any admissible control $u_t^{k}$ together with $u_t^{-k}$ guarantees a unique solution of $x_t^{k}$ on $[0,\, T]$.
\end{defi}

\begin{prob}\label{prob1}
Let the initial condition $x_0$ be given. The objective is to find an optimal control $\tilde{u}_t$ that minimizes ($\star$) where the stochastic dynamics are given by \eqref{eq5a}, and $u_t$ is a control.
\end{prob}

\subsection{Mean field control problem}\label{MFC-sec-prob}
In the mean field control setting, the key feature of the optimum social setting lies in a cooperative game where the optical beam transceiver agents work cooperatively and collectively to attain equivalent decisions by minimizing their own expected social costs and the \textcolor{black}{deviation error performance} from the average sum of other transceiver agents' information and costs. Hence, this optimal tracking control in MFC setting can be viewed as an optimal social type of cooperative game. 
\begin{prob}\label{prob2}
Let the initial condition $x_0$ be given. The objective is to find a solution $\tilde{u}_t = (\tilde{u}_t^1,\cdots,\tilde{u}_t^N)$ with centralized information that minimizes the following social cost function
\begin{equation}\label{eq-social}
\mathcal{J}^N_{\mbox{\scriptsize soc }}(.)=\frac{1}{N}\sum_{k=1}^N \mathcal{J}^k(.)
\end{equation}
where the stochastic dynamics are given by \eqref{eq5a}, $x_0$ is given and the optical transceiver agent $k$ contributes the component~$\mathcal{J}^k$.
\end{prob}

Next, we will derive the solutions of the mean field {\it Problems}~\ref{prob1} and {\it Problem}~\ref{prob2} that minimize the optical transceivers cost ($\star$) and social cost \eqref{eq-social} subject to \eqref{eq5a} through a linear quadratic (LQ) mean field game framework. These solutions result in a different sufficient condition for the unique existence of the underlying equilibrium strategy using the HJB-FP equation method.

\begin{table*}[!t]
\begin{center}
\line(1,0){480}
\end{center}
 \begin{subeqnarray}\label{eq1-mpc11}
\displaystyle\frac{\der\!\eta}{\der\!t}\!&=&\!\Big(\mathbb{A}_t+\Gamma_t-\mathbb{B}_tR_t^{-1}\mathbb{B}_t^\top\phi_t\Big)\eta_t-\mathbb{B}_tR_t^{-1}\mathbb{B}_t^\top\chi_t,\qquad\qquad\qquad\qquad \qquad\quad\eta(0)= \hat{x}_0=\mathbb{E}[x_0] \slabel{eq-mpc11a}\\
-\displaystyle\frac{\der\!\phi}{\der\!t}&=&\mathbb{A}_t^\top \phi_t +\phi_t\mathbb{A}_t-\phi_t\mathbb{B}_tR_t^{-1}\mathbb{B}_t^\top\phi_t+Q_t+\bar{Q}_t,        \qquad\qquad\qquad\qquad\qquad\qquad \phi_T= Q_T+\bar{Q}_T, \slabel{eq-mpc11b}\\
-\displaystyle\frac{\der\!\chi}{\der\!t}&=&\Big(\mathbb{A}_t^\top-\phi_t\mathbb{B}_tR_t^{-1}\mathbb{B}_t^\top\Big) \chi_t+\Big(\phi_t\Gamma_t-\bar{Q}_tS_t\Big)\eta,  \qquad\qquad\qquad\qquad\qquad\quad \chi_T=-\bar{Q}_TS_T\eta_T\slabel{eq-mpc11c}\\
\displaystyle\frac{\der\!\zeta}{\der\!t}&=&\displaystyle \rho\phi_t-\frac{1}{2}(\chi_t)^\top\mathbb{B}_tR_t^{-1}\mathbb{B}_t^\top\chi_t +(\chi_t)^\top\Gamma_t\eta_t+ \frac{1}{2}(\eta_t)^\top S_t^\top\bar{Q}_tS\eta_t,  \qquad\qquad \zeta_T=\frac{1}{2}(\eta_T)^\top S_T^\top\bar{Q}_TS_T\eta_T\slabel{eq-mpc11d}
% \tag{$\star\star$}% \qquad 1\leqslant k\leqslant N
     \end{subeqnarray}
\begin{center}
\line(1,0){480}
\end{center}
\vspace{0.25cm}
\end{table*}

\section{Mean field game and mean field control solutions}\label{sec-solution}
In this section, we provide the solution of the equilibrium and social optima of the optical beam tracking communications \eqref{eq5a} with objectives ($\star$) and \eqref{eq-social} in the mean field game (MFG) and mean field control (MFC), respectively. We focus on the optimal mean field term expressed as the solution of a forward-backward ordinary differential equation using the stochastic maximum principle due to the linearity of the auxiliary equations. This forward-backward structure technique can be generalized to a higher dimension and prevents the time-marching method from solving the system numerically (see, for instance, \cite{LaW:11,Lau:20,HCM:12}). We restrict our attention to the sequel by dropping the superscript $k$ for simplicity.

\subsection{Mean field game solution}\label{MFG-sec-solution}
We state the mean field type (LQ) stochastic game, which provides the optimal feedback control by solving the optimization {\it Problem}\,\ref{prob1} through the HJB-FP equation method.

\begin{prop}[ODE system for MFG solution]\label{thm1} {\it Problem}\,\ref{prob1} is uniquely solvable if there is a unique solution $\tilde{u}_t^{\ast}$, an optimal control $\tilde{u}_t$ and a function $t\rightarrow (\eta_t, \zeta_t,\phi_t, \chi_t)$ that solve the system of ordinary differential equations (ODEs) \eqref{eq-mpc11a}$-$\eqref{eq-mpc11d} where
\begin{equation}\label{eq8a}
\left\{\begin{array}{llllllll}
\displaystyle \tilde{u}_t(x_t)=\frac{1}{2} (x_t)^\top\phi_tx_t+(x_t)^\top\chi_t+\zeta_t\\
\displaystyle\tilde{u}_t^{\ast}(x_t)=-R_t^{-1}\mathbb{B}_t\Big(\phi_t x_t +\chi_t\Big)\\
\displaystyle\frac{1}{N\!-\!1}\!\!\sum_{i=1,i\neq k}^N\!\!\!x_t^{i}=\displaystyle \eta_t \\
\end{array}\right.
\end{equation}
and the minimal expected payoff is 
\begin{equation*}\label{eq9a}
\displaystyle \inf_{\tilde{u}} \mathcal{J}^{\mbox{\scriptsize }}_{\mbox{\scriptsize MFG}}(\tilde{u}^{\ast})=\frac{1}{2}\Big(\var\Big(\sqrt{\phi^0}x_0\Big)+\mathbb{E}\Big[\sqrt{\phi^0}x_0\Big]^2\Big)+x_0^\top\chi_0+\zeta_0
\end{equation*} 
\end{prop}

\subsection{Mean field control solution}\label{MFC-sec-solution}
We derive the mean field type (LQ) control, which provides the optimal feedback control by solving the optimization {\it Problem}\,\ref{prob2} through the HJB-FP equation method.

\begin{prop}[ODE system for MFC solution]\label{thm2}
{\it Problem}\,\ref{prob1} is uniquely solvable if there is a unique solution $\tilde{u}_t^{\star}$, an optimal control $\tilde{u}_t$ and a function $t\rightarrow (\eta_t, \zeta_t,\phi_t, \chi_t)$ that solve the system of ordinary differential equations (ODEs) \eqref{eq-mpc11a}, \eqref{eq-mpc11b} and \eqref{eq-mpc11d} in which \eqref{eq-mpc11c} is replaced by the following equation
\begin{align}\label{eq10a}
&\!\!\!\!\!\!-\displaystyle\frac{\der\!\chi}{\der\!t}=\Big(\mathbb{A}_t^\top-\phi_t\mathbb{B}_tR_t^{-1}\mathbb{B}_t^\top\Big) \chi_t-\bar{Q}_tS_t\eta-S_t^\top\bar{Q}_t\eta \notag \\&\qquad\qquad +\Big(\Gamma_t^\top \phi_t +\phi_t\Gamma_t +S_t^\top\bar{Q}_tS_t\Big)\eta,  \\& \chi_T=-\bar{Q}_TS_T\eta_T
\end{align}
where $\displaystyle\tilde{u}_t^{\star}(x_t)=-R_t^{-1}\mathbb{B}_t\Big(\phi_t x_t +\chi_t\Big)$ is an optimum for the MFC problem, and the minimal expected payoff is
\begin{align*}\label{eq11a}
&\displaystyle \inf_{\tilde{u}} \mathcal{J}^{\mbox{\scriptsize soc}}_{\mbox{\scriptsize MFC}}(\tilde{u}^{\star})=\frac{1}{2}\Big(\var\Big(\sqrt{\phi^0}x_0\Big)+\mathbb{E}\Big[\sqrt{\phi^0}x_0\Big]^2\Big)+x_0^\top \chi_0\notag \\&\qquad\qquad\qquad\qquad+\zeta_0 + \eta_T^\top\Big(I-S_T\Big)\bar{Q}_TS_T\eta_T\notag \\&+\int^T_0 \Big[(\phi_t\eta_t+\chi_t)\Gamma_t\eta_t-\eta_t^\top\Big(I-S_t\Big)\bar{Q}_tS_t\eta_t^\top\Big]\der\!t,
\end{align*} 
where $\mathcal{J}^{\mbox{\scriptsize soc}}_{\mbox{\scriptsize MFC}}(.)=\displaystyle\frac{1}{N}\sum_{k=1}^N \mathcal{J}^k(.)$.
\end{prop}

\begin{proof}
The propositions \ref{thm1} and \ref{thm2} follow the standard results in \cite{BFP:13}. It lies in the stochastic maximum principle and dynamic programming. A detailed exposition of the derivation of these systems of ODEs with the forward-backward method can be found in \cite{BFP:13}. Further, the existence and uniqueness of the ODEs systems \eqref{eq-mpc11a}$-$\eqref{eq-mpc11d} and \eqref{eq10a} for both mean field models can be provided based on the standard Riccati equations \cite{HCM:12,BSYY:16,SMN:20}, or fixed point method \cite{HCM:10}.
\end{proof}
The solution pair $(\eta_t,\chi_t)$ is the population mean distribution and characterizes the best response in these ODEs systems for MFG and MFC models. Equation \eqref{eq-mpc11b} does not depend on the other equations and can be solved independently. The explicit solution of the Riccati equation \eqref{eq-mpc11d} depends on the variables $\eta_t$, $\chi_t$, and $\phi_t$ and admits a unique positive solution. The solution pair to the equations  \eqref{eq-mpc11a} and \eqref{eq-mpc11c} are coupled forward-backward (FB) equations. This FB structure prevents the obstacle of a time-marching approach and can be numerically computed through Newton iterations, and fixed-point iterations \cite{HWW:16,Lau:20}. As a result, we will restrict our attention to the convergence of the solution pair to the two joint equations \eqref{eq-mpc11a} and \eqref{eq-mpc11c}, which correspond to the equilibrium mean and its control part, respectively.%

\begin{remark}
Note that the computation of a Nash equilibrium is mathematically tractable for the mean field optical beam transceiver tracking setting, assuming some design choices of utility functions \cite{DGP:09}. Additionally, the fictitious play method, which remains a breakthrough in finding the Nash equilibrium, can ensure asymptotic convergence through stochastic approximation in this setting \cite{MAS:09}.
\end{remark}

\section{Simulation results}\label{num-simu}
The primary goal of the optical beam tracking control problem is to derive an action profile that is a strategy equilibrium or social optimum. In other words, we exploit the backward-forward framework to derive the solution pair for the two joint equations \eqref{eq-mpc11a} and \eqref{eq-mpc11c} that guarantee a unique Nash equilibrium and social optimum for the MFG and MFC, respectively. Besides, the convergence of this solution pair $(\eta_t,\chi_t)$ is one of the key features of the Nash equilibrium/social optimum for this specific LQ setting. Therefore, we will propose two strategies, fixed point iterations, and Newton iterations, to solve the forward-backward ODEs systems of MFG and MFC frameworks. Finally, we will quantify the PoA that results from the ratio between both settings.

First, we compute the Riccati equation \eqref{eq-mpc11b} as the solution exists and does not depend on other equations. Then, we discretize the system of coupled ODEs \eqref{eq-mpc11a} and \eqref{eq-mpc11c} equations through a finite-difference time scheme with a step size $\Delta t$. Finally, we solve the following linear regression problem
\begin{equation*}\label{eq-POA1}
 \begin{bmatrix} \eta \\ \chi \end{bmatrix}=\Sigma \begin{bmatrix} \eta \\ \chi \end{bmatrix}+\Pi,
\end{equation*}
where $\eta=\begin{bmatrix}\eta_0, \eta_{\Delta t}, \cdots, \eta_{T} \end{bmatrix}^\top$, $\chi=\begin{bmatrix}\chi_0, \chi_{\Delta t}, \cdots,\chi_{T} \end{bmatrix}^\top$ and $\Sigma\in\R^{(V_T+1)\times(V_T+1)}$ and $\Pi\in\R^{(V_T+1)}$ are defined to account for the dynamics, the initial and terminal conditions. $V_T>0$ represents subintervals over the partition interval $[0,\, T]$.
We simulate the performance of the optical beam tracking control in a $1$D angular state (azimuth or elevation) for the transceiver axis and LoS problem with constant parameters. Note that the measurement at each receiver relies on the pointing error of its opposite neighboring transmitter. We consider the following parameters in the numerical simulation.
\begin{align*}
&\mathbb{A}\!=\!\begin{bmatrix} 1.5 & 0 \\ 0 & 1 \end{bmatrix}\!, Q\!=\!\begin{bmatrix} 90 & 0 \\ 0 & 30 \end{bmatrix}\!, \bar{Q}_T\!=\!\begin{bmatrix} 4.5 & 0 \\ 0 & 2.5 \end{bmatrix}\!, \mathbb{B}\!=\!\begin{bmatrix} 1 \\ 0 \end{bmatrix}\!, \\ %T=1,
&S=S_T=\Gamma=Q_T\!=\!\begin{bmatrix} 1 & 0 \\ 0 & 1 \end{bmatrix}, \rho\!=\!\begin{bmatrix} 0.25 & 0 \\ 0 & 0.25 \end{bmatrix}, \\
&\bar{Q}\!=\!\begin{bmatrix} 10 & 0 \\ 0 & 5 \end{bmatrix},  R\!=\!\begin{bmatrix} 130 & 0 \\ 0 & 110 \end{bmatrix}, x_0\!=\!\begin{bmatrix} 40 \\ 20 \end{bmatrix} \, \mbox{and} \, T\!=\!1 {\si s}.
\end{align*}
\subsection{Fixed point and Newton iterations methods}
Numerical methods are crucial in solving MFC and MFG problems. Newton's iteration scheme provides a better tractable way to compute the optimal decentralized strategy. We provide two numerical results based on fixed-point iterations and Newton iterations algorithms to solve the LQ MFC and MFG problems, respectively. Afterward, we quantified the efficiency of the mean field differential games through the price of anarchy performance.

\subsubsection{Fixed point iterations results}
The fixed point iteration method is a recursive approach that solves each discretized ODE alternatively, updates the flow of distributions, and defines a map. The convergence of these iterations holds when the constructed map is a contraction. 
Figs. \ref{fig-1a} and \ref{fig-2a} show the numerical solutions of the MFC using fixed point iterations algorithms of both the transceiver axis and LoS angular states. Figs.~\ref{fig-1a}\textbf{a)} and \ref{fig-2a}\textbf{a)} show the trajectories of the solution pair $(\eta_1,\chi_1)$ and $(\eta_2,\chi_2)$, respectively. Figs.~\ref{fig-1a}\textbf{b)} and \ref{fig-2a}\textbf{b)} show the difference of two consecutive iterates. Hence, the results verify the existence and uniqueness of the mean field equilibrium (MFE) which implies that the control prescribed by the MFE is a social optimum. 

\begin{figure}[!t]
   \begin{minipage}[c]{0.45\linewidth}  
           \centering
      \begin{overpic}[scale=0.3]{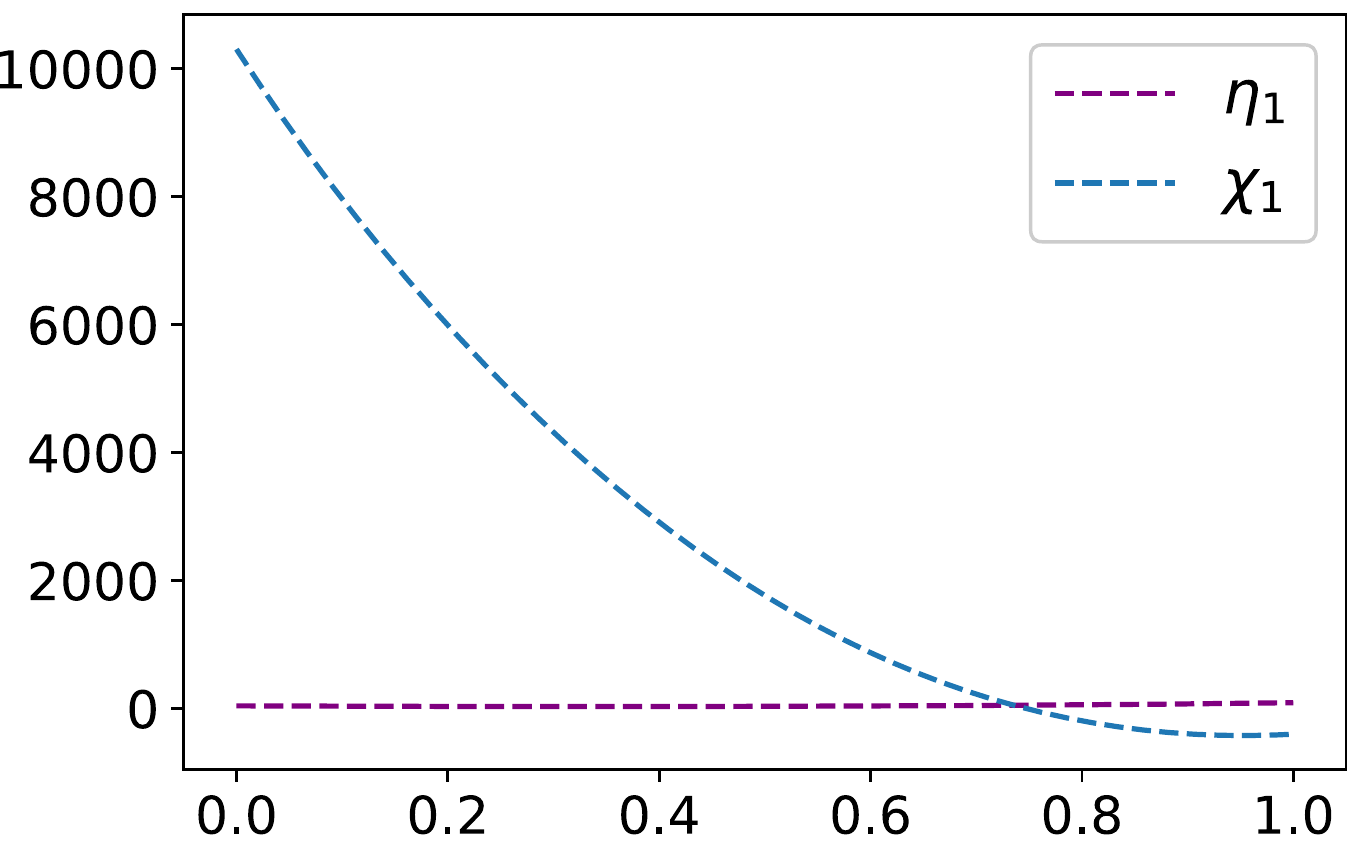}
              \put(44,-4){\scriptsize  Time [{\si sec}]}
      \put(20,-3){\scriptsize  \textbf{a)}}
            \end{overpic}  
       \end{minipage}\hfill 
         \begin{minipage}[c]{0.48\linewidth}
         \centering
      \begin{overpic}[scale=0.3]{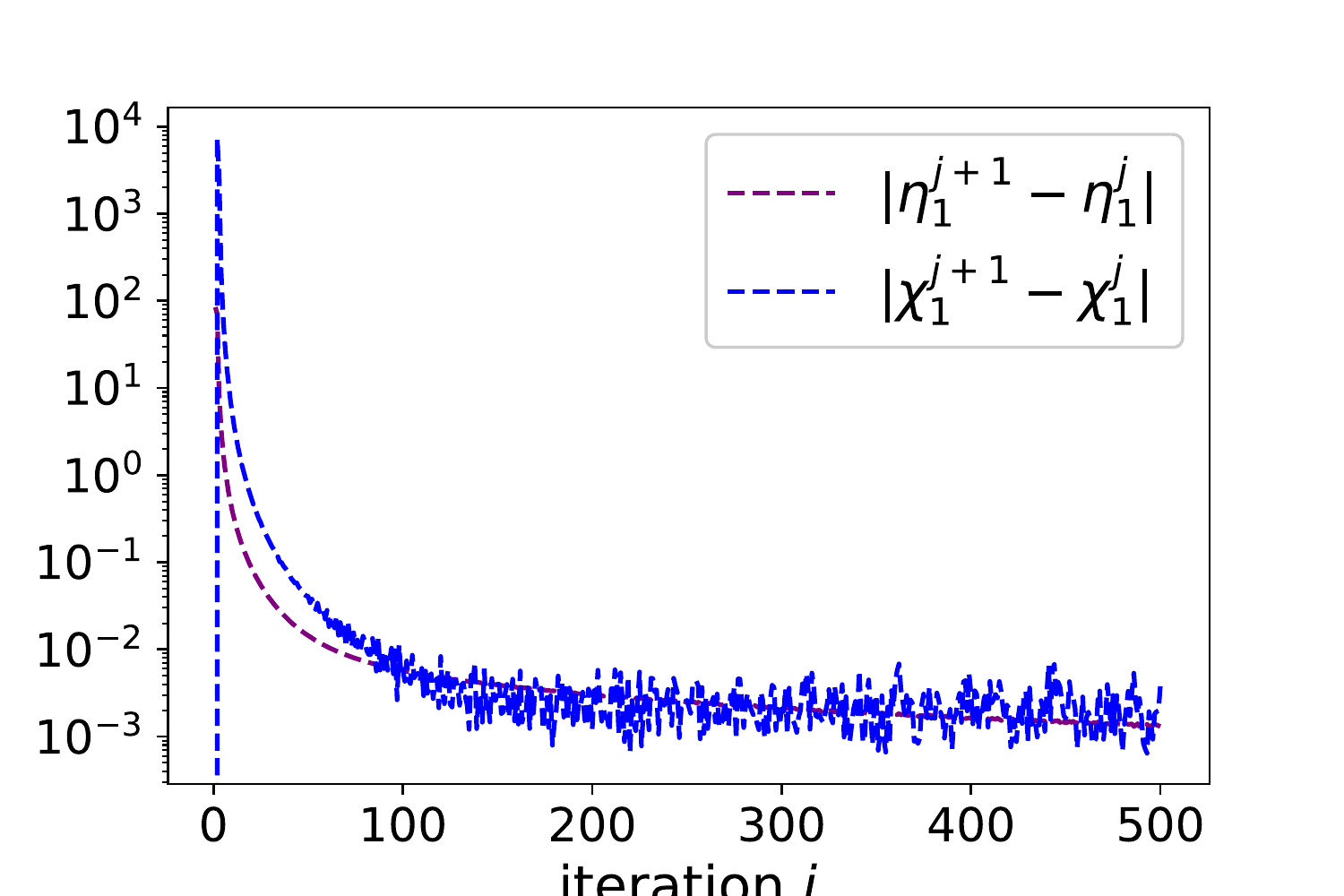}
                \put(44,-4){\scriptsize  Iteration $j$}
      \put(15,-3){ \scriptsize \textbf{b)}}
            \end{overpic}            
      \end{minipage} \vspace{0.2cm}            
      \caption{MFC solution with fixed point iterations including damping for the transceiver state $x^1$: \textbf{a)} solution pair $(\eta_1,\chi_1)$;  \textbf{b)} difference of two consecutive iterates in $\cL2$.} \label{fig-1a}
\end{figure}
\begin{figure}[!t]
   \begin{minipage}[c]{0.45\linewidth}  
           \centering
      \begin{overpic}[scale=0.3]{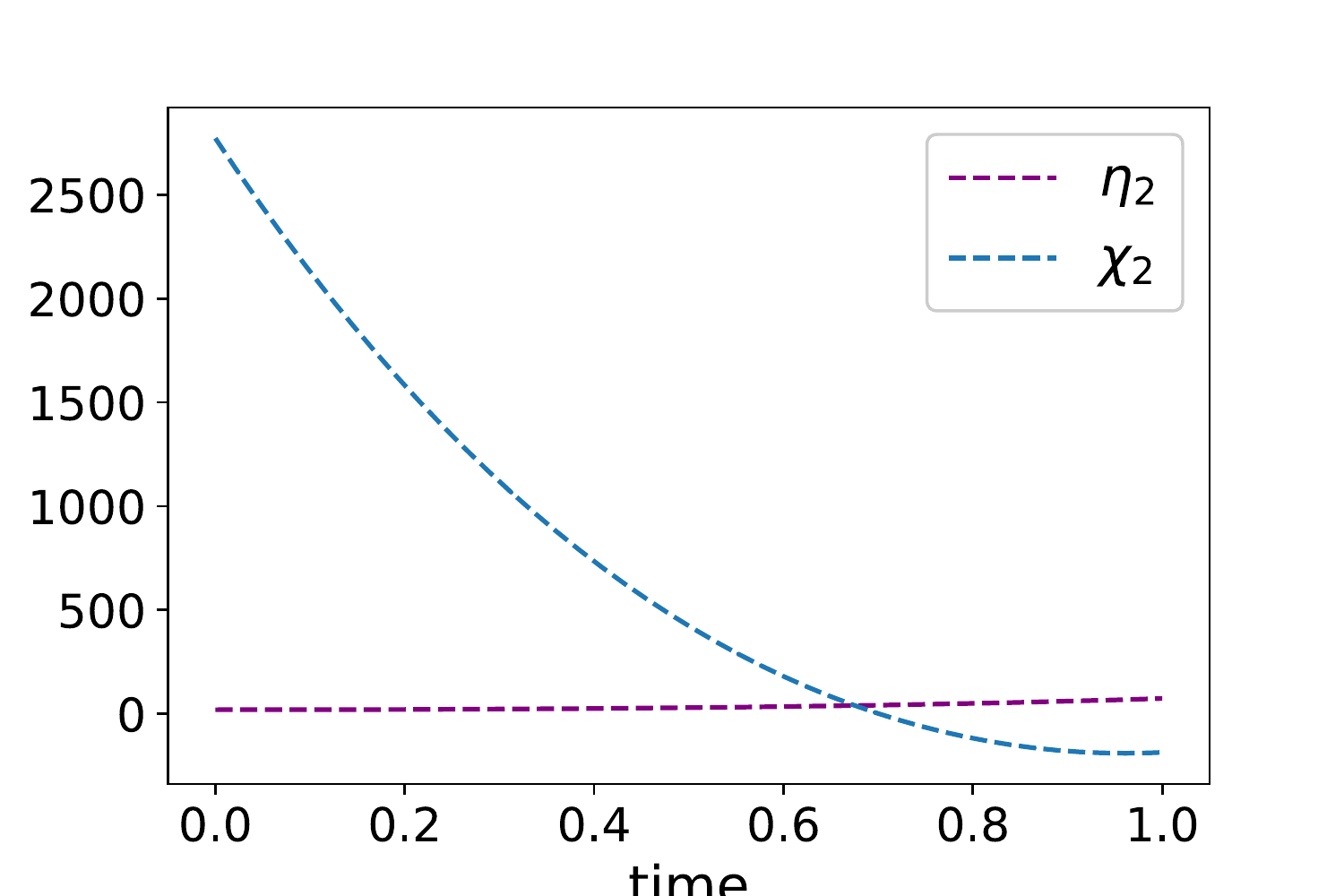}
              \put(44,-4){\scriptsize  Time [{\si sec}]}
      \put(20,-3){\scriptsize  \textbf{a)}}
            \end{overpic}  
       \end{minipage}\hfill 
         \begin{minipage}[c]{0.48\linewidth}
         \centering
      \begin{overpic}[scale=0.3]{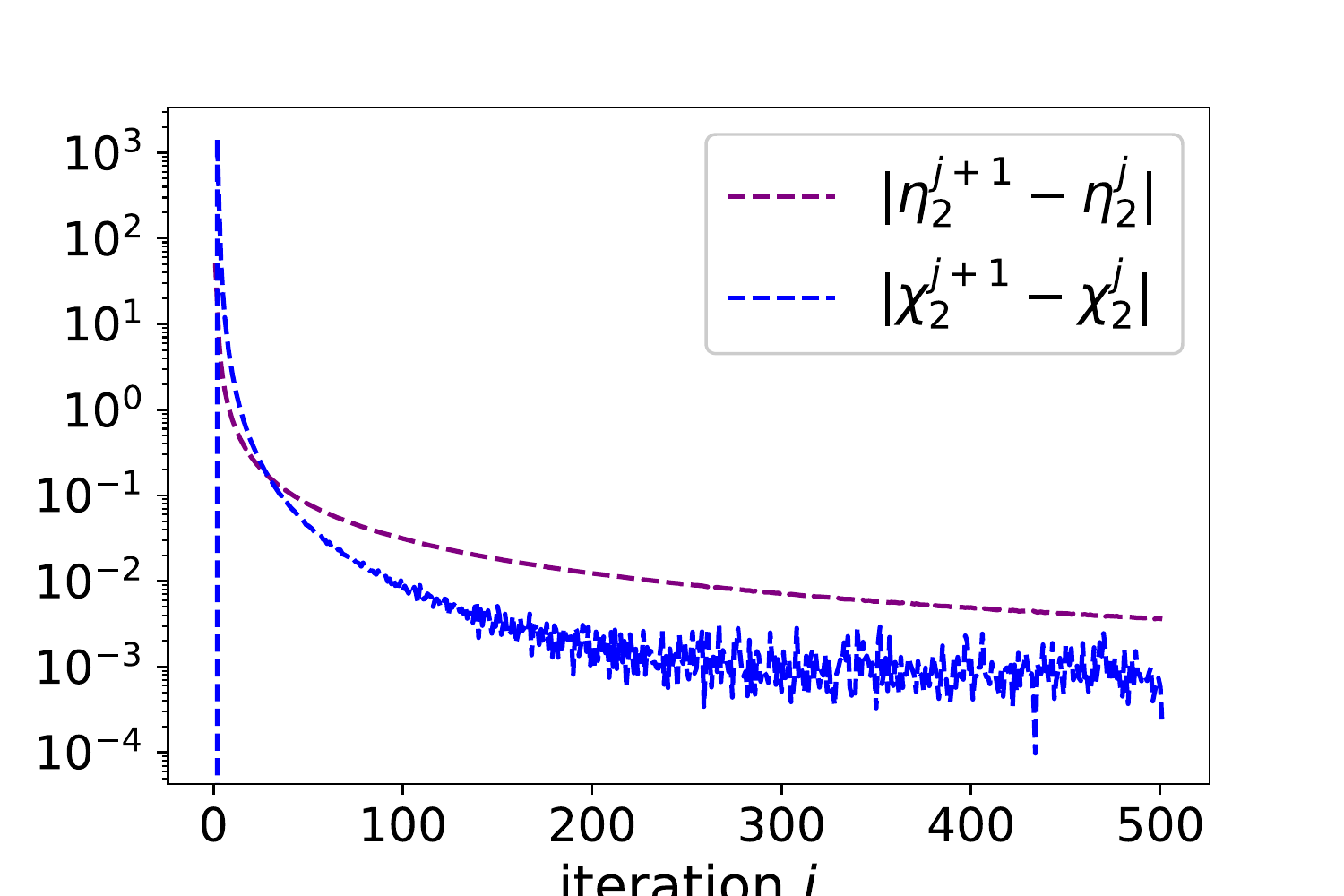}
                \put(44,-4){\scriptsize  Iteration $j$}
      \put(15,-3){ \scriptsize \textbf{b)}}
            \end{overpic}            
      \end{minipage} \vspace{0.2cm}            
      \caption{MFC solution with fixed point iterations including damping for the LoS state $x^2$: \textbf{a)} solution pair $(\eta_2,\chi_2)$;  \textbf{b)} difference of two consecutive iterates in $\cL2$.} \label{fig-2a}
\end{figure}

\subsubsection{Newton iterations results}
Newton's method takes the mean field game system ({\it i.e.,} backward and forward equations) as a single equation and solves this equation, considering good initial guess and terminal conditions \cite{AcP:12, Lau:20}. In the LQ setting, the convergence of this method is guaranteed with a good initial guess. Similarly to the fixed point iterations results for MFC, the existence, uniqueness, and convergence of the mean field equilibrium (MFE) using Newton's method for the MFG differential game are illustrated in Figs.~\ref{Nfig-1a}\textbf{a)} and \ref{Nfig-2a}\textbf{a)} with the solution pair $(\eta_1,\chi_1)$ and $(\eta_2,\chi_2)$, respectively. Figs.~\ref{Nfig-1a}\textbf{b)} and \ref{Nfig-2a}\textbf{b)} show the difference of two consecutive iterates. The results verify the existence and uniqueness of the  MFE which is a Nash equilibrium. 

\begin{figure}[!t]
   \begin{minipage}[c]{0.45\linewidth}  
           \centering
      \begin{overpic}[scale=0.3]{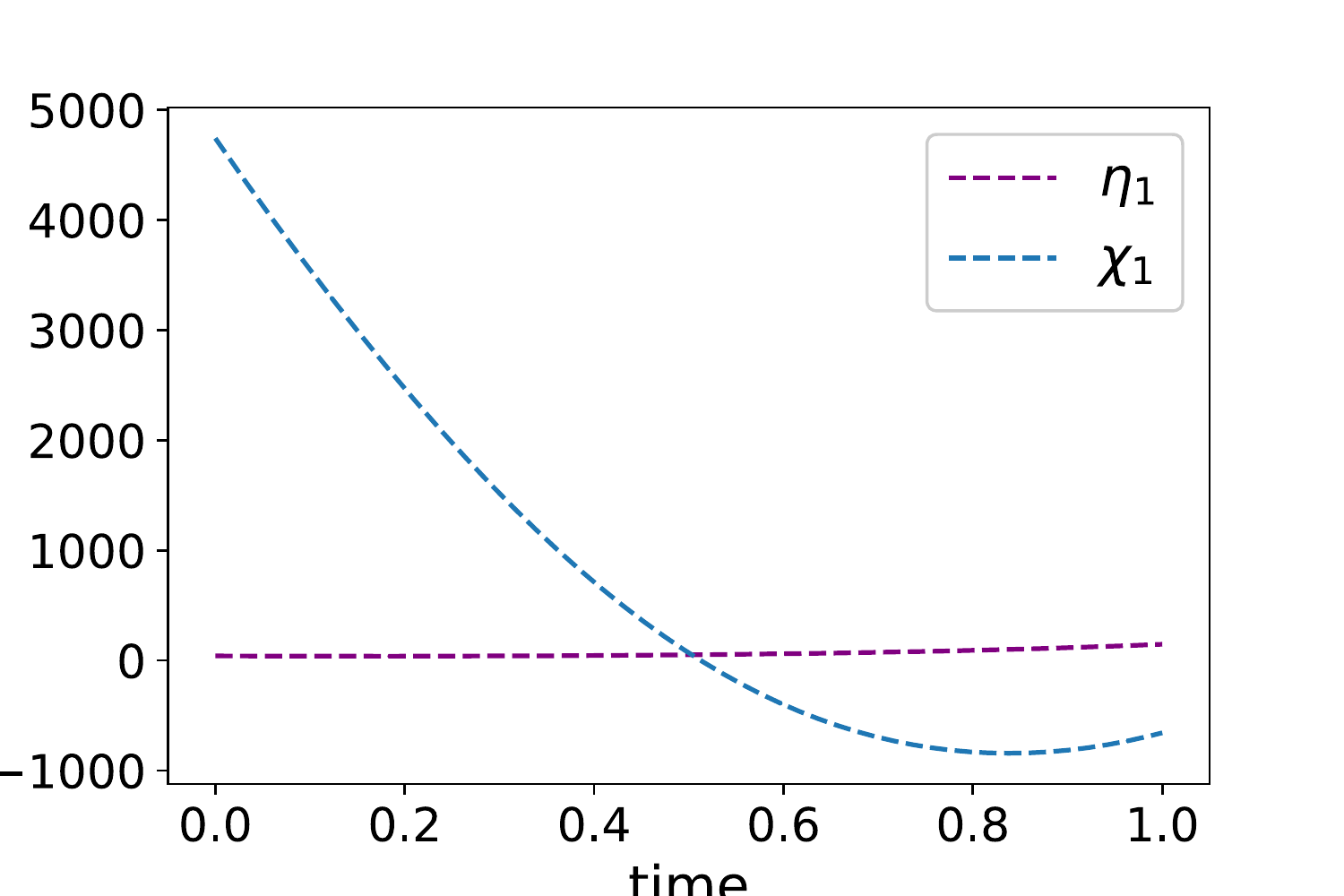}
        \put(44,-4){\scriptsize  Time [{\si sec}]}
      \put(20,-3){\scriptsize  \textbf{a)}}
            \end{overpic}  
       \end{minipage}\hfill 
         \begin{minipage}[c]{0.48\linewidth}
         \centering
      \begin{overpic}[scale=0.3]{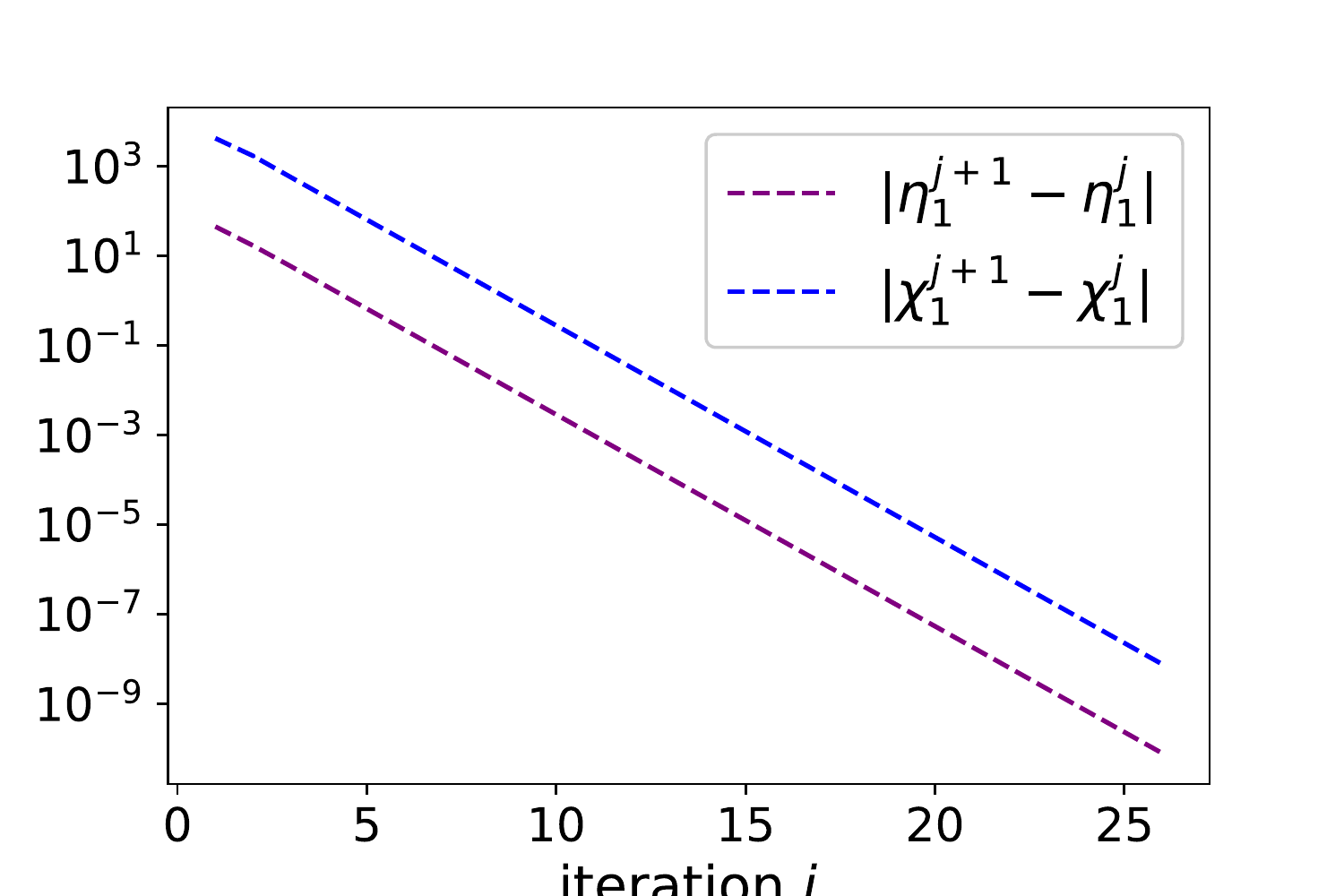}
            \put(44,-4){\scriptsize  Iteration $j$}
      \put(20,-3){ \scriptsize \textbf{b)}}
            \end{overpic}            
      \end{minipage} \vspace{0.2cm}            
      \caption{MFG solution with Newton iterations for the transceiver state $x^1$: \textbf{a)} solution pair $(\eta_1,\chi_1)$;  \textbf{b)} difference of two consecutive iterates in $\cL2$.} \label{Nfig-1a}
\end{figure}
\begin{figure}[!t]
   \begin{minipage}[c]{0.45\linewidth}  
           \centering
      \begin{overpic}[scale=0.3]{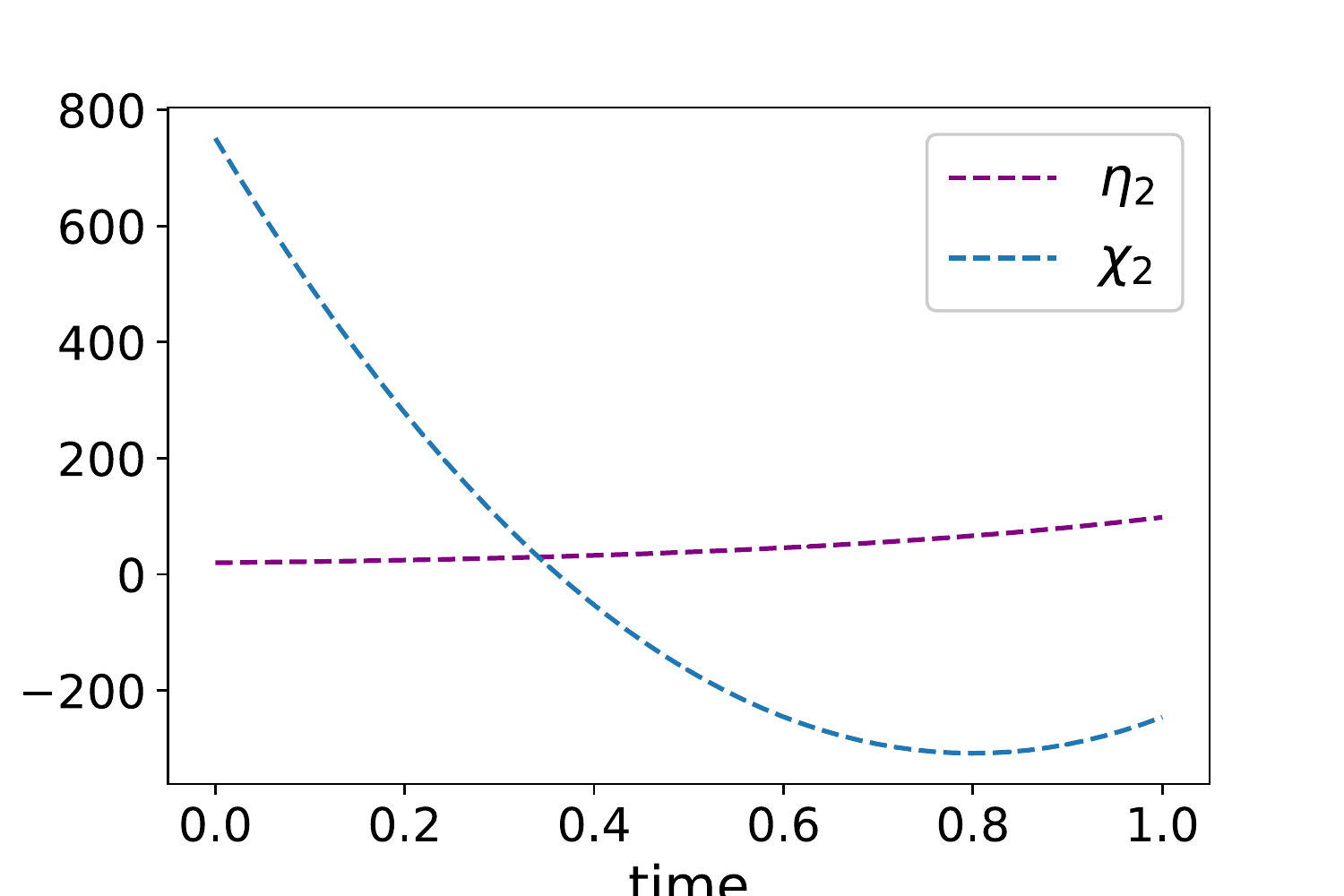}
      \put(20,0){\scriptsize  \textbf{a)}}
      \put(44,-4){\scriptsize  Time [{\si sec}]}
            \end{overpic}  
       \end{minipage}\hfill 
         \begin{minipage}[c]{0.48\linewidth}
         \centering
      \begin{overpic}[scale=0.3]{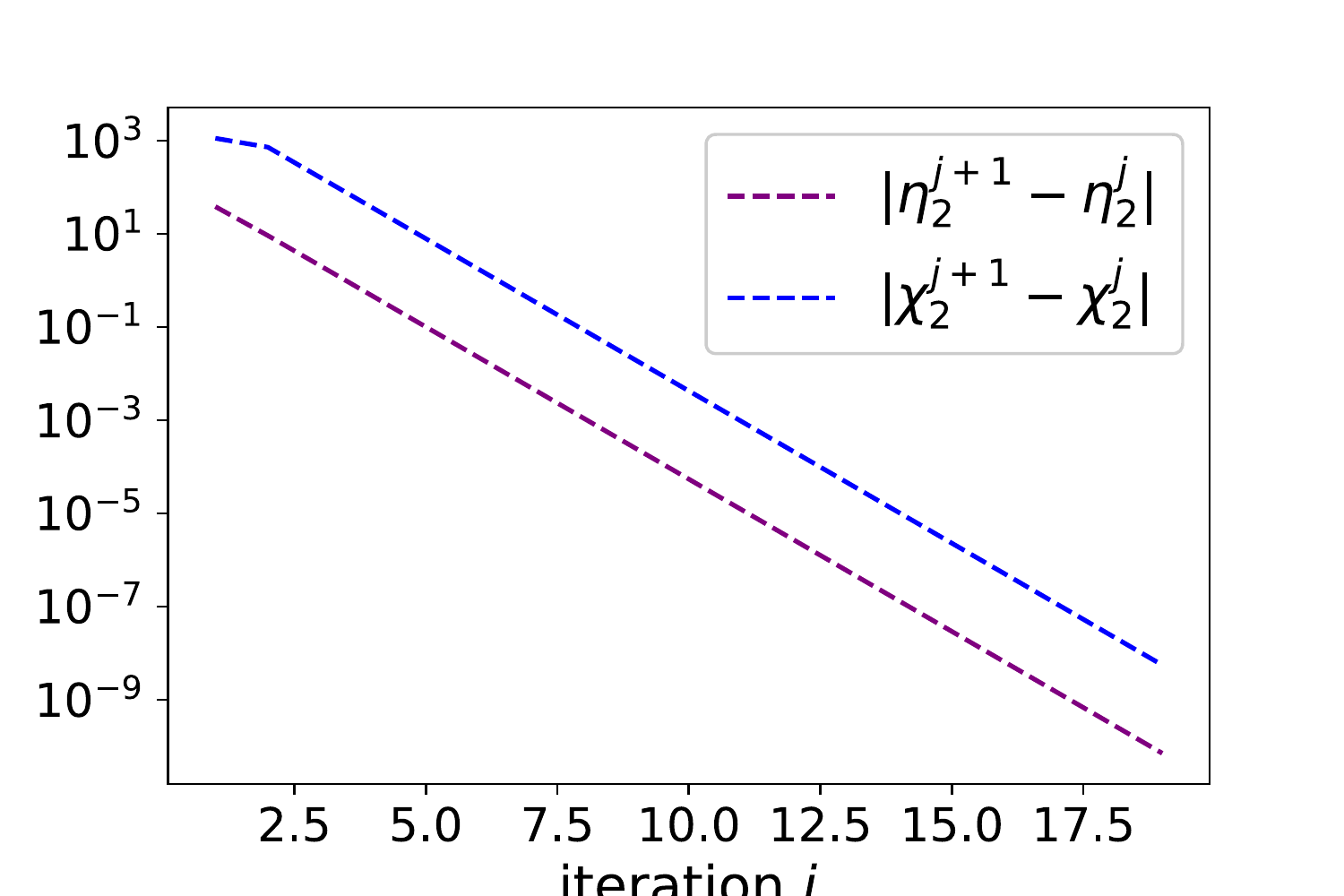}
      \put(44,-4){\scriptsize  Iteration $j$}
      \put(20,-4){ \scriptsize \textbf{b)}}
            \end{overpic}        
      \end{minipage} \vspace{0.2cm}               
      \caption{MFG with Newton iterations for the LoS state $x^2$: \textbf{a)} solution pair $(\eta_2,\chi_2)$;  \textbf{b)} difference of two consecutive iterates in $\cL2$.} \label{Nfig-2a}
\end{figure}

\subsection{Efficiency guarantees: Price of anarchy (PoA)}
The price of anarchy (PoA) is a performance metric that quantifies the efficiency in a mean field game concept \cite{NRTV:07}. It provides 
an upper bound on the ratio between an equilibrium and an optimal allocation's performance \cite{ShA:05,MaS:18}. The price of anarchy (PoA) is defined as
\begin{equation*}\label{eq-POA}
\mbox{PoA}=\displaystyle\frac{\mathcal{J}^{\mbox{\scriptsize}}_{\mbox{\scriptsize MFG}}(\tilde{u}^{\ast})}{\mathcal{J}^{\mbox{\scriptsize soc}}_{\mbox{\scriptsize MFC}}(\tilde{u}^{\star})}.
\end{equation*}
PoA involves the idea of equilibrium and characterizes the performance loss caused by a lack of coordination. PoA is greater than one in the cooperative setting since all the decision-makers minimize the social cost~\cite{BDT:20}. We note from Table \ref{tab:POA-MFCMFG}; the MFE achieves good performances of the price of anarchy in a very short time horizon $T=1{\si s}$ with both iteration algorithms.

\begin{table}[!htb]
\vspace{0.25cm}
    \caption{PoA state average and beam tracking error analysis for MFC and MFG strategies.}\label{tab:POA-MFCMFG}
    \centering
{\small
  \begin{tabular}{|c||cc|}
%\toprule
\hline										
~Price of Anarchy (PoA)	&	~PoA (state-average)~	&	~PoA (tracking)~		\\	\hline \hline%\midrule
~Fixed point  iterations~&	1.0921	&	1.0720		\\	\hline
Newton iterations&	1.0910	&	1.0719	\\	\hline 
\end{tabular}
}
\vspace{0.25cm}
\end{table}

\begin{remark}
It is essential to mention that we can reach the beam tracking error trajectories given the state \eqref{eq5a} and output equation from the results of the controls of both propositions \ref{thm1} and \ref{thm2} in the LQ setting. Further, the computation of the mean field equilibrium and the linearity of the mean field trajectory make the tracking error tractable.
\end{remark}

%%%%%%%%%%%%%%%%%%%%%%%%%%%%%%%%%%%%%%%%%%%%%%%%%%%%%%%%%%%%%%
\section{Conclusion}\label{concl}
This paper proposes two mean field models for optical beam tracking and line-of-sight problems to maintain perfect optical communication alignment between transceivers. We establish two optimal mean field tracking control strategies as the solutions of forward-backward ordinary differential equations relying on the linearity HJB-FP equation structure and the stochastic maximum principle. We numerically validate the existence and uniqueness of the mean field equilibrium (MFE) that drives the control to an MFG equilibrium and MFC social optimum using Newton and fixed-point iterations, respectively. Finally, we quantify the efficiency of the mean field models through the price of anarchy performance.
%the state average and beam tracking error costs and 

%%%%%%%%%%%%%%%%%%%%%%%%%%%%%%%%%%%%%%%%%%%%%%%%%%%%%%%%%%%%%%%%%%%%%%%%%%%%%%%%

\section*{Acknowledgment}
 The authors thank M. Lauri\`ere from New York University (NYU) Shanghai for the insightful course (Mean field games: numerical methods and applications) organized by D. Gomes at the King Abdullah University of Science and Technology (KAUST) and fruitful discussions. 

%\balance
\bibliography{biblio_1}

\end{document}